\documentclass[letterpaper,12pt,oneside,reqno]{amsart}
\usepackage[utf8]{inputenc}%
\usepackage[english]{babel}%
\usepackage{amsmath,amssymb,amsthm,amsfonts}%
\usepackage{hyperref}%
\usepackage{graphicx}
\usepackage{enumerate}
\usepackage[mathscr]{euscript}
\usepackage{color,tikz}
\usepackage[DIV14]{typearea}
\usepackage[width=.9\textwidth]{caption}
\allowdisplaybreaks%
\numberwithin{equation}{section}%

\usepackage{tikz}
\usetikzlibrary{decorations.markings, arrows}

\newcommand{\Z}{\mathbb{Z}}
\newcommand{\C}{\mathbb{C}}
\newcommand{\R}{\mathbb{R}}
\DeclareMathOperator{\E}{\mathbb{E}}
\renewcommand{\i}{\mathbf{i}}
\DeclareMathOperator{\Prob}{\mathsf{Prob}}

\newcommand{\al}{\alpha}

\newcommand{\la}{\lambda}

\newcommand{\be}{\beta}

\DeclareMathOperator*{\Res}{\mathrm{Res}}

\newcommand{\SP}{{{s}}}

\newcommand{\ip}{\xi} 			

\newcommand{\HT}{\mathfrak{h}}

\newcommand{\Sym}{\ensuremath{\mathrm{Sym}}}
\newcommand{\Y}{\ensuremath{\mathbb{Y}}}

\newcommand{\MM}{\ensuremath{\mathbf{MM}}}
\newcommand{\vv}{\ensuremath{\mathbf{6v}}}

\newtheorem{proposition}{Proposition}[section]

\newtheorem{corollary}[proposition]{Corollary}
\newtheorem{theorem}[proposition]{Theorem}
\theoremstyle{definition}
\newtheorem{definition}[proposition]{Definition}
\newtheorem{remark}[proposition]{Remark}
\newtheorem{example}[proposition]{Example}

\begin{document}

\title{Stochastic higher spin six vertex model and Madconald measures}

\author[A. Borodin]{Alexei Borodin}
\address{Department of Mathematics,
Massachusetts Institute of Technology,
77 Massachusetts ave.,
Cambridge, MA 02139, USA\newline
Institute for Information Transmission Problems, Bolshoy Karetny per. 19, Moscow, 127994, Russia}
\email{borodin@math.mit.edu}



\begin{abstract} We prove an identity that relates the q-Laplace transform of the height function of a (higher spin inhomogeneous) stochastic six vertex model in a quadrant on one side, and a multiplicative functional of a Macdonald measure on the other. The identity is used to prove the GUE Tracy-Widom asymptotics for two instances of the stochastic six vertex model via asymptotics analysis of the corresponding Schur measures. 
\end{abstract}

\maketitle

\setcounter{tocdepth}{2}
\tableofcontents
\setcounter{tocdepth}{2}

\section{Introduction} The last two decades have seen a sharp increase in the number of \emph{integrable}, or exactly solvable probabilistic systems. Among others, two fairly general algebraic mechanisms of producing such systems were suggested --- the Macdonald processes \cite{BC}, \cite{BP-lectures}, and the higher spin stochastic six vertex models \cite{CP}, \cite{BP-inhom}, \cite{BP-hom}. The class of Macdonald processes includes an important earlier subclass of the Schur processes \cite{OkR}, \cite{BG-lectures}. 

While there are many similarities between these two mechanisms (they both rely on symmetric functions, both generate interacting particle systems and two-dimensional Markov chains that generalize those, both provide explicit evaluations of averages for broad classes of observables), they appear to be different at the moment.\footnote{Some hope for their unification under a common roof stems from a recent work \cite{GGM}.} 

The goal of this note is to exhibit an identity that relates the q-Laplace transform of the height function of a stochastic six vertex model at a point on one side, and a multiplicative functional of a Macdonald measure on the other; it is stated as Corollary \ref{cor:match} below. The identity involves one free parameter, and comparing Taylor series in this parameter gives countably many identities involving the corresponding moments for the two sides, see Theorem \ref{th:match} below. 

The algebraic nature of this identity remains mysterious to the author. The proof is a direct comparison of previously derived integral representations for both sides. However, the fact that the identity holds with all the essential parameters of the models in the game suggest that there should be a more conceptual proof; it would be very interesting to find one. 

To my best knowledge, the first nontrivial trace of our identity is a result of \cite{ACQ, CDR, D, SS} that can be phrased in the following way: The Laplace transform of the solution of the KPZ (Kardar-Parisi-Zhang) equation with narrow wedge initial data coincides with the expectation of a multiplicative functional of the Airy determinantal random point process, see \cite{BG} for details on this formulation as well as a companion identity of moments. The KPZ-Airy identity can be obtained as a limit of the one from the present work. 

In \cite{IS}, the KPZ-Airy identity was lifted to the level of the O'Connell-Yor semi-discrete Brownian directed polymer. Unfortunately, the associated determinantal point process was not governed by a positive measure.\footnote{Another representation of the Laplace transform of the O'Connell-Yor partition function as the average of a multiplicative functional over a signed determinantal point process can be found in \cite{OC}.} Still, taking the edge limit of this process, the authors were able to recover the KPZ-Airy identity. The O'Connell-Yor polymer can be obtained via an analytic continuation and degeneration of the stochastic six vertex model, cf. \cite[Section 6]{BP-inhom} for a degeneration to the $q$-TASEP, and \cite{BC}, \cite{BP-lectures} for a further degeneration to polymers. Thus, it is possible that our identity would degenerate to the one in \cite{IS}. However, we do not pursue this here as we stay in the realm of the positive measures.

Instead, we focus on how the new identity can be used for asymptotic analysis. 

We show, in Section \ref{sc:equivalence} below, that our identity implies a certain \emph{asymptotic equivalence} of the height function for the stochastic six vertex on one side, and the length (that is, the number of the nonzero parts) of the Macdonald-random partition on the other. One striking consequence is that the asymptotic behavior of the length does not depend on the $(q,t)$-parameters of the Macdonald polynomials for a class of Macdonald measures, see Corollary \ref{cr:equiv-macd} for an exact formulation. One instance of this fact is a recent result \cite[Theorem 1.3]{Dim}, cf. a discussion in \S1.4 there. 

In the same spirit, we prove that in a certain special, yet still fairly general situation, one can replace the Macdonald measures by the Schur measures. The asymptotic analysis of the Schur measures is well developed, see e.g. \cite{BG-lectures} and references therein. An application of this analysis allows us to obtain the result of \cite{BCG} on the GUE Tracy-Widom asymptotics for the height function of the stochastic six vertex model, and also a similar (and new) result for an instance of the higher spin six vertex model, see Section \ref{sc:asymptotics} below.  

The stochastic six vertex has a natural degeneration to the asymmetric simple exclusion process  (ASEP, for short), so it is natural to ask what happens to our identity under such a limit. 
The answer is not entirely trivial and leads outside the class of Macdonald measures; it is presented in \cite{BO} along with other similar results. Another application of the new identity, to Baik-Ben Arous-P\'ech\'e like phase transitions in the stochastic six vertex model, can be found in \cite[Appendix B]{AB}. 

The discussion of this work is related to one-point distributions of generally speaking two-dimensional random fields, and it is natural to ask if any many-point extension exists. At least in one case the answer appears to be positive, but one needs to restrict the class of Macdonald processes to the Hall-Littlewood ones; this will be addressed in \cite{BB}. 

\smallskip

\noindent\textbf{Acknowldegements.}\ I am very grateful to Alexey Bufetov, Ivan Corwin, Vadim Gorin, Gri\-gori Olshanski, and Leonid Petrov for many very helpful discussions. This work was partially supported by the NSF grants DMS-1056390 and DMS-1607901.

\section{Stochastic higher spin six vertex model in a quadrant}\label{sc:6v}

Our exposition in this section largely follows \cite{BP-inhom}.

Consider an ensemble $\mathcal P$ of infinite oriented up-right paths drawn in the first quadrant $\Z_{\ge 1}^2$ of the square lattice, with all the paths starting from a left-to-right arrow entering each of the points $\{(1,m):m\in\Z_{\ge 1}\}$ on the left boundary (no path enters through the bottom boundary). Assume that no two paths share any horizontal piece (but common vertices and vertical pieces are allowed). See Figure \ref{fig:intro}.

\begin{figure}[htb]
	\begin{tikzpicture}
		[scale=.7,thick]
		\def\d{.1}
		\foreach \xxx in {1,...,6}
		{
		\draw[dotted, opacity=.4] (\xxx-1,5.5)--++(0,-5);
		\node[below] at (\xxx-1,.5) {$\xxx$};
		}
		\foreach \xxx in {1,2,3,4,5}
		{
		\draw[dotted, opacity=.4] (0,\xxx)--++(5.5,0);
		\node[left] at (-1,\xxx) {$\xxx$};
		\draw[->, line width=1.7pt] (-1,\xxx)--++(.5,0);
		}
		\draw[->, line width=1.7pt] (-1,5)--++(1-3*\d,0)--++(\d,\d)--++(0,1-\d);
		\draw[->, line width=1.7pt] (-1,4)--++(1-2*\d,0)--++(\d,\d)--++(0,1-2*\d)--++(\d,2*\d)--++(0,1-\d);
		\draw[->, line width=1.7pt] (-1,3)--++(1-\d,0)--++(\d,\d)--++(0,1-2*\d)--++(\d,2*\d)--++(0,1-2*\d)--++(\d,2*\d)--++(0,1-\d);
		\draw[->, line width=1.7pt] (-1,2)--++(1,0)--++(0,1-\d)--++(\d,\d)--++(1-\d,0)
		--++(0,1)--++(2-\d,0)--++(\d,\d)--++(0,2-\d);
		\draw[->, line width=1.7pt]
		(-1,1)--++(3,0)--++(0,2)--++(1,0)--++(0,1-\d)--++(\d,\d)--++(1-\d,0)
		--++(0,2);
		\draw[densely dashed] (-.5,4.5)--++(4,-4) node[above,anchor=west,yshift=16,xshift=-9] {$x+y=5$};
	\end{tikzpicture}
	\caption{A path collection $\mathcal P$.}
	\label{fig:intro}
\end{figure}
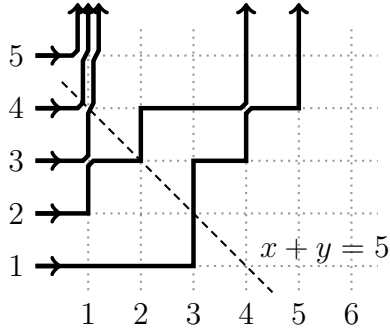

Define a probability measure on the set of such path ensembles in the following Markovian way. For any $n\ge 2$, assume that we already have a probability distribution on the intersections $\mathcal P_n$ of $\mathcal P$ with the triangle $T_n=\{(x,y)\in \Z_{\ge 1}^2: x+y\le n\}$. We are going to increase $n$ by 1. For each point $(x,y)$ on the upper boundary of $T_n$, i.e., for $x+y=n$, every $\mathcal P_n$ supplies us with two inputs: (1) The number of paths that enter $(x,y)$ from the bottom --- denote it by $i_1\in\Z_{\ge 0}$; (2) The number of paths that enter $(x,y)$ from the left --- denote it $j_1\in\{0,1\}$. Now choose, independently for all $(x,y)$ on the upper boundary of $T_n$, the number of paths $i_2$ that leave $(x,y)$ in the upward direction, and the number of paths $j_2$ that leave $(x,y)$ in the rightward direction, using the probability distribution with weights
of the transitions $(i_1,j_1)\to (i_2,j_2)$ given by
\begin{align}\label{intro-weights}
	\begin{array}{rclrcl}
		\Prob((i_1,0)\to (i_2,0))=&\dfrac{1-Q^{i_1} \SP_x \ip_x u_y}{1-\SP_x \ip_x u_y}\,\mathbf 1_{i_1=i_2},\\
		\rule{0pt}{22pt}
		\Prob((i_1,0)\to (i_2,1))=&\dfrac{(Q^{i_1}-1)\SP_x \ip_x u_y}{1-\SP_x \ip_x u_y}\,\mathbf 1_{i_1=i_2+1},\\
		\rule{0pt}{22pt}
		\Prob((i_1,1)\to (i_2,1))=&\dfrac{Q^{i_1} \SP_x^{2}-\SP_x \ip_x u_y}{1-\SP_x \ip_x u_y}\,\,\mathbf 1_{i_1=i_2},\\
		\rule{0pt}{22pt}
		\Prob((i_1,1)\to (i_2,0))=&\dfrac{1-Q^{i_1} \SP_x^{2}}{1-\SP_x \ip_x u_y}\,\mathbf 1_{i_1=i_2-1}.
	\end{array}
\end{align}

Assuming that all above expressions are nonnegative, this procedure defines a probability measure on the set of all $\mathcal P$'s because we always have $\sum_{i_2,j_2} \Prob((i_1,j_1)\to (i_2,j_2))=1$, and $\Prob((i_1,j_1)\to (i_2,j_2))$ vanishes unless $i_1+j_1=i_2+j_2$.

To ensure the nonnegativity of the right-hand sides of \eqref{intro-weights}, we will use the following assumptions on the parameters:
\begin{itemize}
\item  $0<Q<1$;
\item $\xi_x,u_y>0$ {  for all  } $x,y\ge 1$;
\item for any $x\ge 1$, either $s_x=Q^{-m/2}$ for some $m=1,2,\dots$, or $s_x\in (-1,0)$;
\item $\xi_xu_y>s_x$ for any $x,y\ge 1$ (this is trivially satisfied if $s_x\in(-1,0)$).  
\end{itemize}

Observe that if $s_x=Q^{-m/2}$ then $\Prob((m,1)\to (m+1,0))=0$, which means that no more than $m$ paths can share the same vertical piece in the column located at $x$. The case of $m=1$ (that is, no two paths can share an edge) corresponds to the stochastic six vertex model introduced in \cite{GwaSpohn} and recently studied in \cite{BCG}.

Each path ensemble $\mathcal P$ can be encoded by a \emph{height function}
$\mathfrak h:\Z_{\ge 1}\times\Z_{\ge 1} \to \Z_{\ge 0}$, that assigns to each 
vertex $(M,N)$ the number $\mathfrak h(M,N)$ of paths in $\mathcal P$ that pass 
through or to the right of this vertex.
The value $\HT(M,N)$ clearly depends only on the behavior of the paths in the 
$(M-1) \times N$ rectangle formed by $(M-1)$ first columns and $N$ first rows. 
It is convenient to introduce a notation for the data that describes 
$\mathcal P$ in such a rectangle. 

\begin{definition}\label{df:6v-spec} A collection 
$\mathcal S^{\vv} = \bigl(M,N,\{s_x\}_{x=1}^{M-1}, 
\{\xi_x\}_{x=1}^{M-1},\{u_y\}_{y=1}^N\bigr)$
is called a \emph{specification} of the stochastic higher spin six vertex in a quadrant with parameter $Q\in(0,1)$, if $M,N\ge 1$, and the $s,\xi$, and $u$ parameters satisfy the above nonnegativity conditions.
\end{definition}

\begin{proposition}\cite[Lemma 9.11]{BP-inhom} Take $Q\in (0,1)$ and assume 
that we are given a specification $\mathcal S^{\vv}$ in the sense of Definition 
\ref{df:6v-spec}. Then for any $\ell\ge 1$

\begin{multline}\label{intro-multi_moments}
		\E \prod_{i=1}^{\ell}
		\big(Q^{\HT(M,N)}-Q^{i-1}\big)
		=
		Q^{\frac{\ell(\ell-1)}2}
		\oint\frac{d w_1}{2\pi\i}
		\ldots
		\oint\frac{d w_\ell}{2\pi\i}
		\prod_{1\le a<b\le 
\ell}\frac{w_a-w_b}{w_a-Qw_b}
		\\\times
		\prod_{i=1}^{\ell}\bigg(
		w_i^{-1}
		\prod_{x=1}^{M-1}
		\frac{1-\SP_x\xi_x^{-1}w_i}{1-\SP_x^{-1}\xi_x^{-1}w_i}
		\prod_{y=1}^{N}\frac{1-Qu_yw_i}{1-u_yw_i}
		\bigg),
	\end{multline}
where the integration contours are sufficiently small positively oriented 
simple curves that encircle the points $w_*=u_y^{-1}$, $y=1,\dots,N$. 
\end{proposition}
Note that the other potential poles $w_*=(s_x\xi_x)^{-1}$ of the integrand 
always lie outside the integration contours because they are either negative 
(and $u_y>0$ for any $y$), or $(s_x\xi_x)^{-1}<s_x\xi_x^{-1}<u_y$ for any $x$ 
and $y$ due to our assumptions. The presence or absence of the potential poles 
$w_a=Qw_b$ inside the integration contours is irrelevant due to vanishing of 
the corresponding residues, cf. \cite[proof of Theorem 8.13]{BP-inhom}.

Evaluating the right-hand side of \eqref{intro-multi_moments} as a sum of 
residues, we see that \eqref{intro-multi_moments} is actually an identity of 
rational functions in participating parameters, with the caveat that the 
left-hand side may not always be interpreted as an expectation over a positive 
measure. 

\section{Macdonald measures}\label{sc:Macdonald}

Our exposition in this section follows \cite[Section 2]{BC}. Our notation for partitions,
Macdonald symmetric functions, etc. is mostly the standard one used in \cite{Macdonald1995}. 

Let $\Y$ be the set of all partitions and $\Sym$ be the algebra of symmetric functions in indeterminates $x_1,x_2,\dots$ with coefficients in $\C(q,t)$. Here $q$ and $t$ are (in the general case formal) parameters that we will assume to belong to $[0,1)$. A particularly nice linear basis of $\Sym$ is formed by the Macdonald symmetric functions $P_\la(x_1,x_2,\dots;q,t)$ indexed by $\la\in \Y$. 
The Macdonald symmetric functions are orthogonal with respect to a dot product on $\Sym$ defined in terms of the power sums via
\begin{equation*}
\langle p_{\lambda},p_{\mu}\rangle = \langle p_{\lambda},p_{\mu}\rangle_{q,t} = \delta_{\lambda \mu}z_{\lambda}(q,t),\qquad z_{\lambda}(q,t)= z_{\lambda} \prod_{i=1}^{\ell(\lambda)}\frac{1-q^{\lambda_i}}{1-t^{\lambda_i}}, \qquad z_{\lambda} = \prod_{i\geq 1}i^{m_{i}}(m_i)!,
\end{equation*}
for $\lambda=1^{m_1}2^{m_2}\cdots$. Along with $P_{\lambda}$ one defines
\begin{equation*}
Q_{\lambda} = \frac{P_{\lambda}}{\langle P_{\lambda},P_{\lambda}\rangle},
\end{equation*}
so that $P_{\lambda}$ and $Q_{\mu}$ are orthonormal.

Specializing $q=t$ recovers the Schur symmetric functions $s_\la(x)$ that are independent of the parameters, $q=0$ recovers the Hall-Littlewood symmetric functions with parameter $t$, and taking $q=t^{\alpha}$ with $t\rightarrow 1$ recovers the Jack symmetric functions with parameter $\alpha$.

The complete homogeneous symmetric function $h_{r}$ has a $(q,t)$-analog which is denoted $g_r=Q_{(r)}$ and can be expressed as
$g_r = \sum_{|\lambda|=r} z_{\lambda}(q,t)^{-1}p_{\lambda}$ (this is analogous in the sense that $h_{r} = s_{(r)}$). The $g_r$'s with $r=1,2,\dots$ form an algebraically independent system of generators for $\Sym$.

The Macdonald symmetric polynomials are defined as restrictions of the $P_{\lambda}$'s to finitely many variables $x_1,\ldots, x_m$ and written as $P_{\lambda}(x_1,\ldots, x_m)$. If $m<\ell(\lambda)$ then $P_{\lambda}(x_1,\ldots, x_m)=0$.

For any two sets of indeterminates $x_1,x_2,\ldots$ and $y_1,y_2,\ldots$ define
\begin{equation}\label{PiDef1}
\Pi(x;y)=\sum_{\lambda\in \Y} P_{\lambda}(x) Q_{\lambda}(y).
\end{equation}
Then \cite[VI,(2.5)]{Macdonald1995},
\begin{equation}\label{eqn12}
\Pi(x;y)= \prod_{i,j} \frac{(tx_i y_j;q)_{\infty}}{(x_i y_j;q)_{\infty}}= \exp\left(\sum_{n\geq 1} \frac{1}{n} \frac{1-t^n}{1-q^n} p_n(x)p_n(y)\right),
\end{equation}
where $(a;q)_{\infty}=(1-a)(1-aq)(1-aq^2)\cdots$ is the $q$-Pochhammer symbol. This is known as the \emph{Cauchy identity} for Macdonald symmetric functions. 

If the $P_{\lambda}$ and $Q_{\lambda}$ are considered as symmetric functions in variables $x_1,x_2,\ldots$ and $y_1,y_2,\ldots$ (respectively) then the Cauchy identity holds in the sense of formal power series.  If either side is an absolutely convergent series, then the identity turns into a numeric equality. 

Let us introduce a convenient extension to the concept of evaluating at a sequence of variables. 

A {\it specialization}\index{specialization} $\rho$ of $\Sym$ is an algebra homomorphism of $\Sym$ to $\C$. We denote the application of $\rho$ to $f\in\Sym$ as $f(\rho)$. For two specializations $\rho_1$ and $\rho_2$ we define their union $\rho=(\rho_1,\rho_2)$ as the specialization defined on power sum symmetric functions via
\begin{equation*}
p_n(\rho_1,\rho_2)=p_n(\rho_1)+p_n(\rho_2), \qquad n\ge 1,
\end{equation*}
and extended to $\Sym$ by linearity. 

\begin{definition}
A specialization $\rho$ of $\Sym$ is {\it Macdonald nonnegative} (or just `nonnegative') if it takes nonnegative values on the skew Macdonald symmetric functions: $P_{\lambda/\mu}(\rho)\ge 0$ for any partitions $\la$ and $\mu$\footnote{The skew functions $P_{\la/\mu}$ turn into the ordinary $P_\la$ when $\mu=\varnothing$; we did not define them as they play no role in what follows.}.
\end{definition}

There is no known classification of the Macdonald nonnegative specializations. The classification is known in the case of nonnegative specializations of the Jack symmetric functions \cite{KOO}, and in the subcase of Schur symmetric functions this is a classical statement known as \emph{Thoma's theorem}, see \cite{Ker}, \cite{BO-book}, and references therein. In the Macdonald case, however, it is not hard to come up with a class of examples. In fact, Kerov conjectured that this class completely classifies all nonnegative specializations (\cite{Ker}, section II.9). Let us describe it.  

Let $\{\alpha_i\}_{i\ge 1}$, $\{\beta_i\}_{i\ge 1}$, and $\gamma$ be nonnegative numbers, and $\sum_{i=1}^\infty(\alpha_i+\beta_i)<\infty$. Let $\rho$ be a specialization of $\Sym$ defined by
\begin{equation}\label{tag1}
\sum_{n\ge 0} g_n(\rho) u^n= \exp(\gamma u) \prod_{i\ge 1} \frac{(t\alpha_iu;q)_\infty}{(\alpha_i u;q)_\infty}\,(1+\beta_i u)=: \Pi(u;\rho).
\end{equation}
Since $g_n$ form an algebraically independent system of generators of $\Sym$, this uniquely defines the specialization $\rho$.
As zero $\alpha$'s and $\beta$'s do not change this expression, we will simply assume that all $\alpha_i$ and $\beta_i$ are strictly positive (there may be finitely many of them, or none at all). 

The middle expression in (\ref{tag1}) can be viewed as a specialization of \eqref{eqn12}. More generally, for any two specializations  $\rho_1,\rho_2$ set
\begin{equation}\label{PIeqn}
\Pi(\rho_1;\rho_2)=\sum_{\la\in\Y}P_\la(\rho_1)Q_\la(\rho_2)= \exp\left(\sum_{n\ge 1}\frac{1}{n}\,\frac{1-t^n}{1-q^n}\,p_n(\rho_1)p_n(\rho_2)\right)
\end{equation}
provided that the series converge. One can show that for nonnegative specializations $\rho_1$, $\rho_2$ of the form \eqref{tag1}, $\Pi(\rho_1;\rho_2)<\infty$ if and only if the product of any $\alpha$-parameter of $\rho_1$ and any $\alpha$-parameter of $\rho_2$ is $<1$. 

\begin{definition}\label{macprocessdef}
For any two nonnegative specializations $\rho_1,\rho_2$ such that $\Pi(\rho_1;\rho_2)<\infty$, define the {\it Macdonald measure} $\MM(\rho_1,\rho_2)$ as the probability measure on $\Y$ that assigns to a partition $\la\in\Y$ the weight
\begin{equation*}
\MM(\rho_1;\rho_2)(\la)=\frac{P_\la(\rho_1)Q_\la(\rho_2)}
{\Pi(\rho_1;\rho_2)}\,. 
\end{equation*}
\end{definition}

Consider a Macdonald measure $\MM(\rho_1,\rho_2)$ with $\rho_1=(x_1,\dots,x_n)$, a specialization into $n$ positive variables, and $\rho_2$ of the form \eqref{tag1} with $\gamma=0$. For finiteness of the measure we need to assume that $x_i\alpha_j<1$ for any $i,j\ge 1$. Note that such a measure is supported by partitions $\la$ with $\ell(\la)\le n$ because otherwise $P_\la(x_1,\dots,x_n)$ vanishes. 

\begin{proposition}\label{pr:macd-exp} For any $0\le \ell\le n$ we have
\begin{multline}\label{eq:macd-exp}
\E e_\ell(q^{\la_1} t^{n-1},q^{\la_2}t^{n-2},\dots,q^{\la_n})\\ =\frac {1}{(2\pi \i)^\ell \ell!}\oint\cdots\oint
\det\left[\frac 1{tz_a-z_b}\right]_{a,b=1}^\ell
\prod_{i=1}^\ell  \left(\prod_{m=1}^n \frac{tz_i-x_m}{z_i-x_m}\prod_{j\ge 1}\frac{1-\alpha_jz_i}{1-t\alpha_jz_i}\frac{1+q\beta_jz_i}{1+\beta_jz_i}\right)
dz_j,
\end{multline}
where the integration contours are sufficiently small positively oriented simple curves encircling the points $z_*=x_j$, $j=1,\dots,n$, and $e_\ell$'s are the elementary symmetric polynomials.\footnote{Recall that $e_\ell (y_1,\dots,y_n)=\sum_{1\le i_1<\dots<i_\ell\le n}y_{i_1}\cdots y_{i_n}$.}
\end{proposition}

Note that, similarly to \eqref{intro-multi_moments}, the other potential poles $(t\alpha_j)^{-1}$ and $-\beta_j^{-1}$ always lie outside the integration contours as they are either negative, or $(t\alpha_j)^{-1}>x_i$ for any $i,j$, because we assumed that $x_i\alpha_j<1$ and $t<1$. Also similarly to \eqref{intro-multi_moments}, the presence of potential poles $tz_a=z_b$ inside the integration contours is irrelevant due to vanishing of the corresponding residues. 

Proposition \ref{pr:macd-exp} is a straightforward corollary of (2.32), Proposition 2.2.9 and Proposition 2.2.11 of \cite{BC}.

\begin{remark}\label{rm:independence} If there are no $\beta_j$'s present, the 
right-hand side of \eqref{eq:macd-exp} is manifestly independent of $q$, which 
means that the expectation in the left-hand side does not depend on $q$ either. 

Similarly, if the $\beta_j$'s are present, replacing $q$ by $\tilde q= q^{1/k}$ 
and $\{\beta_j\}$ by $$
\{\tilde\beta_j\}=\{\beta_j\}\sqcup\{q^{1/k}\beta_j\}\sqcup\dots\sqcup 
\{q^{(k-1)/k}\beta_j\}$$ 
with $k$ being an arbitrary integer $\ge 1$, also does not change either side 
of \eqref{eq:macd-exp}. 

\end{remark}

As for the six vertex model, it will be convenient for us to collect the data that gives rise to a Macdonald measure used in \eqref{eq:macd-exp} into a single notation. 

\begin{definition}\label{df:mac-spec} A collection $\mathcal S^\MM=\{n,\{x_i\}_{i=1}^n, \{\alpha_i\}_{i\ge 1}, \{\beta_i\}_{i\ge 1}\}$ is called a \emph{specification} of the Macdonald measure if $n\ge 1$, the sets $\{\alpha_i\}_{i\ge 1}$, $\{\beta_i\}_{i\ge 1}$ are finite, all participating parameters $x_i,\alpha_i,\beta_i$ are positive, and $x_i\alpha_j<1$ for any $i,j$. \footnote{Note that the convergence condition $\sum_i(\alpha_i+\beta_i)<\infty$ is automatically satisfied.}
\end{definition} 

\section{Matching expectations}\label{sc:match}
The goal of this section is to provide conditions on specifications of the higher spin six vertex model (see Definition \ref{df:6v-spec}) and the Macdonald measure (see Definition \ref{df:mac-spec}) that would imply the coincidence of \eqref{intro-multi_moments} and \eqref{eq:macd-exp}. 

\begin{definition}\label{df:match} We say that a specification $\mathcal S^{\vv}$ of the higher spin six vertex model in a quadrant with parameter $Q$ of Definition \ref{df:6v-spec} \emph{matches} a specialization $\mathcal S^{MM}$ of the Macdonald measure with parameters $(q,t)$ of Definition \ref{df:mac-spec}, if the following conditions are satisfied: 
\begin{itemize}
\item $Q=t$, $N=n$, $\{u_1,\dots, u_N\}=\{x_1^{-1},\dots,x_N^{-1}\}$. 
\item There exists a splitting of the set $\{\alpha_i\}$ into clusters forming geometric progressions of ratio $t$:
$$
\{\alpha_i\}_{i\ge 1}=\bigsqcup_{j\ge 1} C_{k_j,t}(\tilde\alpha_j),\qquad C_{k,t}(\tilde\al)=\{\tilde\al,t\tilde\al,\dots,t^{k-1}\tilde\al\},
$$
a splitting of the set $\{\beta_i\}$ into clusters forming geometric progressions of ratio $q$:
$$
\{\beta_i\}_{i\ge 1}=\bigsqcup_{j\ge 1} C_{l_j,q}(\tilde\beta_j),\qquad C_{l,q}(\tilde\be)=\{\tilde\be,q\tilde\be,\dots,q^{l-1}\tilde\be\},
$$

a bijection
$$
\{C_{k_i,t}(\tilde\alpha_i)\}_{i\ge 1}\sqcup \{C_{l_j,q}(\tilde\beta_j)\}_{j\ge 
1} \longleftrightarrow
\{s_x\}_{x=1}^{M-1} ,
$$
such that clusters of $\alpha$'s correspond to positive $s_x$'s, and clusters of $\beta$'s correspond to negative $s_x$'s. 

\item If in the above bijection a cluster $C_{k,t}(\tilde\alpha)$ corresponds to some $s_x$, $1\le x\le M-1$, then $s_x=t^{-k/2}=Q^{-k/2}$ and $\xi_x=t^{-k/2}\tilde\alpha^{-1}$. 

\item If, on the other hand, in the above bijection a cluster $C_{l,q}(\tilde\beta)$ corresponds to some $s_x$, then $s_x=-q^{l/2}$ and $\xi_x=q^{-l/2}\tilde\beta^{-1}$. 
\end{itemize}
\end{definition}

This peculiar definition is justified by the following statement. 

\begin{theorem}\label{th:match} Assume that specifications $\mathcal S^{\vv}$ and $\mathcal S^{\MM}$ match in the sense of Definition \ref{df:match}. Then for any $0\le \ell\le N$ 
\begin{equation}\label{eq:match1}
{(-1)^\ell}\E_{\vv} \prod_{i=1}^{\ell}\frac{
		Q^{\HT(M,N)}-Q^{i-1}}{1-Q^{i}}=\E_\MM e_\ell(q^{\la_1} t^{n-1},q^{\la_2}t^{n-2},\dots,q^{\la_n}).
\end{equation}
\end{theorem}
\begin{proof} We need to match the right-hand sides of \eqref{intro-multi_moments} and \eqref{eq:macd-exp}. To do that we use \cite[Proposition 3.2.2]{BC} that gives (with a slight change in notation) for any continuous function $f$
\begin{multline*}
\frac{(-1)^\ell}{(2\pi \i)^\ell} \oint \cdots \oint \prod_{1\leq a<b\leq \ell} \frac{w_a-w_b}{w_a-Qw_b}\,\prod_{j=1}^{\ell}  \frac{f(w_j)dw_j}{w_j}
\\= \frac{Q^{\frac{-\ell(\ell-1)}{2}}(Q;Q)_\ell}{(2\pi \i)^{\ell}\ell!} \oint \cdots \oint \det\left[\frac{1}{Qz_i-z_j}\right]_{a,b=1}^{\ell} \prod_{j=1}^{\ell} f(z_j) dz_j,
\end{multline*}
where the $z_j$-contours and $w_j$-contours are all the same (this identity is proved by a straightforward symmetrization of the integration variables). This provides a match for the cross-terms in the two integrals, and it remains to compare the multiplicative terms. 

For a cluster $C_{k,t}(\tilde \alpha)$ we have
$$
\prod_{\al\in C_{k,t}(\tilde \al)} \frac{1-\al z}{1-t\al z}=\frac{1-\tilde \alpha z}{1-t^k\tilde\alpha z}=
\frac{1-s \xi^{-1}z}{1-s^{-1}\xi^{-1}z}\quad \textrm{for}\quad s=t^{-k/2},\ \xi=t^{-k/2}\tilde\alpha^{-1}. 
$$
On the other hand, for a cluster $C_{l,q}(\tilde\beta)$ we have 
$$
\prod_{\be\in C_{l,q}(\tilde \be)}\frac{1+q\be z}{1+\be z}=\frac{1+q^l\tilde\beta z}{1+\tilde\beta z}=\frac{1-s \xi^{-1}z}{1-s^{-1}\xi^{-1}z}\quad \textrm{for}\quad s=-q^{l/2},\ \xi=q^{-l/2}\tilde\beta^{-1}. 
$$
We thus see that the two integrands completely coincide, and so do the 
integration contours. 
\end{proof}

\begin{example}\label{ex:homog} The homogeneous stochastic six vertex model with $s_x\equiv s=Q^{-1/2}$, $\xi_x\equiv 1$, $u_y\equiv u$, corresponds to the Macdonald measure with $t=Q$, $x_i\equiv x= u^{-1}$, $\alpha_j\equiv \alpha= t^{-1/2}$. The positivity condition $u>s=Q^{-1/2}$ for the former 
exactly translates into the convergence condition $x\alpha=u^{-1}Q^{-1/2}<1$ for the latter. 
Note that $q$ here can be arbitrary, cf. Remark \ref{rm:independence}. In particular, one can take $q=t$, which turns the measure on partitions into a Schur measure and removes its dependence on $q$ and $t$ (dependence on $q$ and $t$ remains in the observables). 

The homogeneous stochastic six vertex model with $s_x\equiv s\in (-1,0)$, 
$\xi_x\equiv 1$, $u_y\equiv u$, corresponds to the Macdonald measure with $t=Q$, 
$q=s^2$, $x_i\equiv x= u^{-1}$, $\beta_j\equiv \beta= q^{-1/2}$. This time both 
parameters $q$ and $t$ are uniquely determined, and to see a Schur measure we must have $s=-Q^{1/2}$. 
\end{example}

By multiplying both sides of \eqref{eq:match1} by $\zeta^l$, summing over $0\le \ell\le N$, and using the q-binomial theorem, we also obtain 
\begin{corollary}\label{cor:match} Assume that specifications $\mathcal S^{\vv}$ and $\mathcal S^{\MM}$ match in the sense of Definition \ref{df:match}. Then we have the following equality of polynomials in $\zeta$:
\begin{equation}\label{eq:match2}
\E_{\vv} \prod_{i\ge 0} \frac{1+\zeta Q^i}{1+\zeta Q^{\HT(M,N)+i}}  =\E_{\MM} \prod_{j=1}^N (1+\zeta q^{\la_j}t^{N-j}).  
\end{equation}
Equivalently, for any $\zeta\notin -Q^{\Z_{\le 0}}$, 
\begin{equation}\label{eq:match3}
\E_{\vv} \prod_{i\ge 0} \frac{1}{1+\zeta Q^{\HT(M,N)+i}}  =\E_{\MM} \prod_{j\ge 0} \frac{1+\zeta q^{\la_{N-j}}t^{j}}{1+\zeta t^j}\,,
\end{equation}
where in the right-hand side we assume that $q^{\la_{-m}}=0$ for $m\ge 0$.  
\end{corollary} 

The advantage of \eqref{eq:match3} over the equivalent polynomial identity \eqref{eq:match2} is that both observables in \eqref{eq:match3} take values between 0 and 1, which will become useful in the next section. 

\section{Asymptotic equivalence}\label{sc:equivalence}
The goal of this section is to extract asymptotic information about underlying probability measures from the observables in \eqref{eq:match3}. 

\begin{definition}\label{df:spread} Let $\{\eta_n\}_{n\ge 1}$ be a sequence of real-valued random variables. We say that this sequence \emph{spreads} as $n\to\infty$ if 
$$
\lim_{n\to\infty}\sup_{x\in\R} \Prob\{x<\eta_n\le x+1\}=0.
$$ 
Equivalently, the above condition says that the chance of finding $\eta_n$ in an interval of given (finite) length goes to zero as $n\to\infty$, uniformly in the location of the interval. 

We also say that a sequence $\{F_n(x)\}_{n\ge 1}$ of non-decreasing functions $F_n:\R\to \R$ \emph{spreads} if 
$$
\lim_{n\to\infty}\sup_{x\in\R} (F_n(x+1)-F_n(x))=0.
$$
Clearly, $\{\eta_n\}_{n\ge 1}$ spreads if and only if the corresponding sequence of cumulative distribution functions $\{F_{\eta_n}(x)=\Prob\{\eta_n\le x\}\}_{n\ge 1}$ spreads. 

The definition naturally extends to families indexed by more general index sets with a well-defined notion of a limiting point (for example for an index $a\in\R$, $n\to\infty$ can be replaced by $a\to +\infty$). \footnote{A more formal definition could be given in terms of convergent filter bases, but we won't need this level of generality.}
\end{definition}

\begin{definition}\label{df:equiv} Two sequences $\underline{\eta}=\{\eta_n\}_{n\ge 1}$ and $\underline{\zeta}=\{\zeta_n\}_{n\ge 1}$ of real-valued random variables are said to be  \emph{asymptotically equivalent} if 
\begin{itemize}
\item $\underline{\eta}$ spreads if and only if $\underline{\zeta}$ spreads;
\item assuming $\underline{\eta}$ and $\underline{\zeta}$ spread, 
$$
\lim_{n\to\infty} \sup_{x\in\R}\,(\Prob\{\eta_n\le x\}-\Prob\{\zeta_n\le x\})=0. 
$$
\end{itemize} 

Similarly, a sequence $\underline{\eta}=\{\eta_n\}_{n\ge 1}$ of real-valued random variables and a sequence $\underline F=\{F_n(x)\}_{n\ge 1}$ of non-decreasing functions $F_n:\R\to \R$ are \emph{asymptotically equivalent} if 
\begin{itemize}
\item $\underline{\eta}$ spreads if and only of $\underline F$ spreads;
\item assuming $\underline{\eta}$ and $\underline{F}$ spread, 
$$
\lim_{n\to\infty} \sup_{x\in\R}\,(\Prob\{\eta_n\le x\}-F_n(x))=0. 
$$
\end{itemize} 
This definition also extends to more general index sets with a notion of the limiting point. 
\end{definition} 

\begin{proposition}\label{pr:equiv} Let $\{\eta_n\}_{n\ge 1}$ be a sequence of real-valued random variables, and let $\{\phi_{n,x}\}_{n\ge 1, x\in \R}$ be another family of real-valued random variables, with $\phi_{n,x}$ defined on the same probability space $\Omega_n$ as $\eta_n$. Assume that 
\begin{enumerate}
\item $0\le \phi_{n,x}\le \phi_{n,y}\le 1$ for any $n\ge 1$ and $x\le y$.
\item If $\eta_n-x\to +\infty$, then $\phi_{n,x}\to 0$, uniformly in $n\ge 1$ and $x\in \R$. More formally, for any $\varepsilon>0$ there exists $M>0$ such that on $\{(x,\omega)\in\R\times\Omega_n: \eta_n(\omega)-x>M\}$ we have $\phi_{n,x}(\omega)<\varepsilon$.
\item  If $\eta_n-x\to -\infty$, then $\phi_{n,x}\to 1$, uniformly in $n\ge 1$ and $x\in \R$.
\item There exists an independent of $n$ constant $c>0$ such that on $\{(x,\omega)\in\R\times\Omega_n:x<\eta_n\le x+1\}$ we have $\phi_{n,x+1}(\omega)-\phi_{n,x}(\omega)\ge c$. 
\end{enumerate}
Then the sequences of random variables $\{\eta_n\}_{n\ge 1}$ and of non-decreasing functions $\{F_n(x):=\E \phi_{n,x}\}_{n\ge 1}$ are asymptotically equivalent. 
\end{proposition}

\begin{remark}\label{rm:as} As the above statement deals with random variables, all the above conditions need to be understood in the almost sure context. However, if we think of $\{\eta_n\}$ and $\{\phi_{n,x}\}$ as of everywhere defined functions on $\Omega_n$ and assume that the conditions are satisfied everywhere, not just up to measure zero subsets, then these conditions will remain satisfied for any probability distributions on $\Omega_n$'s (because they simply do not depend on those). 
\end{remark}

\begin{example}\label{ex:function} Here is one situation when the assumptions of the above proposition hold. 
Let $\Phi:\R\to\R$ be a continuous, strictly increasing function such that $\lim_{x\to-\infty}\Phi(x)=0$ and $\lim_{x\to+\infty} \Phi(x)=1$. Let $\{\eta_n\}_{n\ge 1}$ be an arbitrary sequence of random variables. Define $\phi_{n,x}=\Phi(x-\eta_n)$ for all $n\ge 1$ and $x\in\R$. Then one readily sees that all the assumptions of Proposition \ref{pr:equiv} are satisfied (the last one follows from the fact that $\min_{-1\le y\le 0}(\Phi(y+1)-\Phi(y))>0$ due to strict monotonicity). Hence, $\{\eta_n\}_{n\ge 1}$ and $\{F_n(x)=\E \Phi(x-\eta_n)\}_{n\ge 1}$ are asymptotically equivalent. Note that this leads to a nontrivial conclusion only when $\{\eta_n\}_{n\ge 1}$ spreads as $n\to\infty$. 
\end{example}

The above example is similar to \cite[Lemma 4.1.39]{BC}. We will also see other, more involved examples later in this section. 

\begin{proof}[Proof of Proposition \ref{pr:equiv}] Assume that $\underline\eta=\{\eta_n\}_{n\ge 1}$ spreads. For any $x\in\R$, $M>0$ we write
\begin{multline*}
F_n(x+1)-F_n(x)=\E(\phi_{n,x+1}-\phi_{n,x})=\int_{\omega:\eta_n(\omega)\le -M+x} (\phi_{n,x+1}(\omega)-\phi_{n,x}(\omega))d\omega\\+\int_{\omega:-M+x<\eta_n(\omega)\le M+x} (\phi_{n,x+1}(\omega)-\phi_{n,x}(\omega))d\omega+\int_{\omega:M+x<\eta_n(\omega)} (\phi_{n,x+1}(\omega)-\phi_{n,x}(\omega))d\omega.
\end{multline*}
According to the assumptions, we can find $M>0$ such that the first and the third integrals are small, uniformly in $n$ and $x$. The fact that $\underline\eta$ spreads implies that for a fixed $M$ and large enough $n$, $\Prob\{-M+x<\eta_n(\omega)\le M+x\}$ is arbitrarily small (uniformly in $x$), which leads to the smallness of the second integral due to the boundedness of $\phi_{n,x}$. 
Hence, the whole expression can be made small uniformly in $x$ by choosing an appropriate $M$ and large enough $n$, and this means that $\underline F=\{F_n(x)\}_{n\ge 1}$ spreads. 

Assume that $\underline\eta$ does not spread. This means that for any large enough $n$ there exists $x(n)\in\R$ such that $\Prob\{x(n)< \eta_n\le x(n)+1\}\ge c'>0$. Then, using the last assumption, 
\begin{multline*}
\E(\phi_{n,x(n)+1}-\phi_{n,x(n)})\ge \int_{\omega:x(n)<\eta_n(\omega)\le x(n)+1} (\phi_{n,x(n)+1}(\omega)-\phi_{n,x(n)}(\omega))d\omega\\ \ge c \Prob\{x(n)< \eta_n\le x(n)+1\}\ge cc'>0,
\end{multline*}
which means that $\underline F$ also does not spread. 

Assume now that both $\underline\eta$ and $\underline F$ spread. Then writing 
\begin{equation*}
F_n(x)=\E\phi_{n,x}=\int_{\omega:\eta_n(\omega)\le -M+x} \phi_{n,x}+\int_{\omega:-M+x<\eta_n(\omega)\le M+x} \phi_{n,x}+\int_{\omega:M+x<\eta_n(\omega)} \phi_{n,x}
\end{equation*}
and using the assumptions, we can choose $M>0$ that makes the third term arbitrarily small, and the first term arbitrarily close to $\Prob\{\eta_n\le -M+x\}$. On the other hand, the fact that $\underline\eta$ spreads implies that $\Prob\{\omega:-M+x<\eta_n(\omega)\le M+x\}\to 0$ as $n\to\infty$. This shows that for a certain choice of $M$ and sufficiently large $n$, $|F_n(x)-\Prob\{\eta_n\le x\}|$ is arbitrarily small uniformly in $x\in\R$, as required. 
\end{proof}

\begin{proposition}\label{pr:column} Take $\Omega_n=\{\la\in\Y:\ell(\la)\le n\}$ and $q\in[0,1)$, $t\in(0,1)$. Then 
$$
\eta_n:\Omega_n\to \Z_{\le 0},\quad \eta_n:\la\mapsto \ell(\la)-n, \qquad\qquad \phi_{n,x}:\Omega_n\to\R, \quad \phi_{n,x}:\la\mapsto \prod_{j\ge 0} \frac{1+q^{\la_{n-j}}t^{j+x}}{1+t^{j+x}}\,,
$$
where $\ell(\la)$ is the number of nonzero parts in $\la$ (the \emph{length} of $\la$) and $q^{\la_{-m}}=0$ for $m\ge 0$, satisfy the conditions of Proposition \ref{pr:equiv} (see also Remark \ref{rm:as}). Consequently, 
$$
\{\eta_n=\ell(\la)-n\}_{n\ge 1} \qquad \textrm{and}\qquad \left\{F_n(x)=\E \prod_{j\ge 0} \frac{1+q^{\la_{n-j}}t^{j+x}}{1+t^{j+x}}\right\}_{n\ge 1}
$$
are asymptotically equivalent as $n\to\infty$ in the sense of Definition \ref{df:equiv}, for an arbitrary choice of probability distributions on the $\Omega_n$'s. 
\end{proposition} 
\begin{proof} When $q=0$ we have 
$$
\prod_{j\ge 0} \frac{1+q^{\la_{n-j}}t^{j+x}}{1+t^{j+x}}=\prod_{j\ge 0} \frac{1}{1+t^{j+x+n-\ell(\la)}}=\prod_{j\ge 0} \frac{1}{1+t^{(x-\eta_n)+j}},
$$
and the result follows from Example \ref{ex:function}. 

Assume that $q>0$. Let us check the conditions of Proposition \ref{pr:equiv} one by one. 

The inequalities $0\le \phi_{n,x}\le \phi_{n,y}\le 1$ for $x\le y$ hold, because each factor in the definition of $\phi_{n,x}$ is a non-decreasing function in $x$. 

If $\eta_n-x=\ell(\la)-n-x>M$ then 
\begin{multline*}
 \prod_{j\ge 0} \frac{1+q^{\la_{n-j}}t^{j+x}}{1+t^{j+x}}=\prod_{j\ge n-\ell(\la)} \frac{1+q^{\la_{n-j}}t^{j+x}}{1+t^{j+x}}= \prod_{j\ge 0} \frac{1+q^{\la_{\ell(\la)-j}}t^{j-\ell(\la)+n+x}}{1+t^{j-\ell(\la)+n+x}}\le \prod_{j=0}^M \frac{1+qt^{j-\ell(\la)+n+x}}{1+t^{j-\ell(\la)+n+x}}\\=
 \prod_{j=0}^M \frac{t^{-j+\ell(\la)-n-x}+q}{t^{-j+\ell(\la)-n-x}+1}\le \left(\frac{1+q}{2}\right)^M
\end{multline*}
because $(y+q)/(y+1)\le (1+q)/2$ for $0\le y\le 1$. This implies the second condition. 

If $\eta_n-x=\ell(\la)-n-x<-M$ then 
$$
\prod_{j\ge 0} \frac{1+q^{\la_{n-j}}t^{j+x}}{1+t^{j+x}}=\prod_{j\ge n-\ell(\la)} \frac{1+q^{\la_{n-j}}t^{j+x}}{1+t^{j+x}}\ge \prod_{j\ge 0} \frac{1}{1+t^{j-\ell(\la)+n+x}}>
\prod_{j\ge 0} \frac{1}{1+t^{j+M}}\,,
$$
and the last expression clearly converges to $1$ as $M\to \infty$. This implies the third condition. 

Finally, for $x<\ell(\la)-n\le x+1$ and $\ell(\la)>0$ we have
$$
\phi_{n,x}=\prod_{j\ge 0} \frac{1+q^{\la_{\ell(\la)-j}}t^{j-\ell(\la)+n+x}}{1+t^{j-\ell(\la)+n+x}}\ge \prod_{j\ge 0} \frac{1}{1+t^{j-\ell(\la)+n+x}}\ge \prod_{j\ge 0} \frac{1}{1+t^{j-1}}=const>0,
$$
and 
\begin{multline}\label{eq:ineq}
\frac{\phi_{n,x+1}}{\phi_{n,x}}=\prod_{j\ge n-\ell(\la)} \frac{1+q^{\la_{n-j}}t^{j+x+1}}{1+t^{j+x+1}}\left(\frac{1+q^{\la_{n-j}}t^{j+x}}{1+t^{j+x}}\right)^{-1}\\ \ge \frac{1+q^{\la_{\ell(\la)}}t^{-\ell(\la)+n+x+1}}{1+t^{-\ell(\la)+n+x+1}}\left(\frac{1+q^{\la_{\ell(\la)}}t^{-\ell(\la)+n+x}}{1+t^{-\ell(\la)+n+x}}\right)^{-1}=\frac{1+q^{\la_{\ell(\la)}}t^{y}}{1+t^{y}}\left(\frac{1+q^{\la_{\ell(\la)}}t^{y-1}}{1+t^{y-1}}\right)^{-1},
\end{multline}
where $y=-\ell(\la)+n+x+1\in [0,1)$, and we used the fact that each factor in the definition of $\phi_{n,x}$ is a non-decreasing function of $x$. Since for $\ell(\la)\ge 1$, 
$$
\frac{1+q^{\la_{\ell(\la)}}t^{y}}{1+t^{y}}=q^{\la_{\ell(\la)}}+\frac{1-q^{\la_{\ell(\la)}}}{1+t^{y}}
$$ 
is actually a strictly increasing function of $y$, the last expression of \eqref{eq:ineq} is bounded from below by a constant that is strictly greater than 1 for $y\in [0,1)$ . This implies that ${\phi_{n,x+1}}-{\phi_{n,x}}$ is uniformly bounded from below by a positive constant when evaluated on any nonempty partition. For the empty partition, we have $\ell(\la)=0$, thus $x<\ell(\la)-n\le x+1$ reads $x<-n\le x+1$, and 
\begin{multline*}
{\phi_{n,x+1}}(\varnothing)-{\phi_{n,x}}(\varnothing)=\prod_{j\ge n}\frac 1{1+t^{j+x+1}}-\prod_{j\ge n}\frac 1{1+t^{j+x}}\\=\left(1-\frac1{1+t^{n+x}}\right)\prod_{j\ge 0}\frac 1{1+t^{j+n+x+1}}\ge const>0, 
\end{multline*}
where the constant does not depend on $x\in [-n-1,-n)$. 

Thus, we have verified all of the assumptions of Proposition \ref{pr:equiv}, and the result follows.
\end{proof}

In the case $q=t$, the function $\phi_{n,x}$ from Proposition \ref{pr:column} can be written as
$\prod_{i\in I}(1+q^{x+i})^{-1}$, where $I=\Z_{\ge 0}\setminus \{\la_n, \la_{n-1}+1,\dots,\la_1+(n-1)\}$; note that $\min(I)=n-\ell(\la)$. With this form of $\phi_{n,x}$, the result and its proof are actually independent of the nature of $I\subset\Z_{\ge 0}$. Let us state the corresponding claim separately as it will be useful in \cite{BO}. 

\begin{corollary} Take $q\in(0,1)$ and let $\{J_n\}_{n\ge 1}$ be a sequence of random subsets of $\Z_{\ge 0}$ (equivalently, a sequence of simple random point processes on $\Z_{\ge 0}$). Define 
$$
\eta_n=-\min (J_n), \qquad \phi_{n,x}=\prod_{j\in J_n}\frac 1{1+q^{x+j}}\,,\qquad n\ge 1,\ x\in \R. 
$$
Then these random variables satisfy the assumptions of Proposition \ref{pr:equiv}, and hence the sequences $\{\eta_n=-\min (J_n)\}_{n\ge 1}$ and $\{F_n(x)=\E \prod_{j\in J_n}(1+q^{x+j})^{-1}\}_{n\ge 1}$ are asymptotically equivalent as $n\to\infty$ in the sense of Definition \ref{df:equiv}. 
\end{corollary}

The proof is very similar to that of Proposition \ref{pr:equiv} (and coincides with it for $J_n$ of the form $\Z_{\ge 0}\setminus \{\la_n, \la_{n-1}+1,\dots,\la_1+(n-1)\}$ with a random $\la\in\Y$ of length $\le n$), and we omit it. 

The reader may have noticed that we have excluded $t=0$ from the statement of Proposition \ref{pr:column}, the reason being that $\phi_{n,x}$ in that case makes little sense. There is, however, a slightly different family of $\phi_{n,x}$ that captures the asymptotic behavior of $\ell(\la)$ in a similar fashion. 

\begin{proposition}\label{pr:t0} Take $\Omega_n=\{\la\in\Y:\ell(\la)\le n\}$ and $q\in[0,1)$. Then 
$$
\eta_n:\Omega_n\to \Z_{\le 0},\quad \eta_n:\la\mapsto \ell(\la)-n, \qquad\qquad \phi_{n,x}:\Omega_n\to\R, \quad \phi_{n,x}:\la\mapsto \prod_{0\le j< -x} {q^{\la_{n-j}}},
$$
where an empty product is assumed to be equal to 1, satisfy the conditions of Proposition \ref{pr:equiv} (see also Remark \ref{rm:as}). Consequently, 
$$
\{\eta_n=\ell(\la)-n\}_{n\ge 1} \qquad \textrm{and}\qquad \biggl\{F_n(x)=\E \prod_{0\le j< -x} q^{\la_{n-j}}\biggr\}_{n\ge 1}
$$
are asymptotically equivalent as $n\to\infty$ in the sense of Definition \ref{df:equiv}, for an arbitrary choice of probability distributions on the $\Omega_n$'s. 
\end{proposition}

\begin{proof} 
If $\eta_n> x$ then in the product $\prod_{0\le j< -x} q^{\la_{n-j}}$, a total of $[\eta_n-x]$ factors corresponding to 
$$
j=n-\ell(\la),n-\ell(\la)+1,\dots,n-\ell(\la)+[\eta_n-x]-1
$$ 
all contribute nontrivial powers of $q$, which means that the product uniformly converges to 0 as $\eta_n-x\to+\infty$.

If $\eta_n<x$ then $n-\ell(\la)>-x$, and the product $\prod_{0\le j< -x} q^{\la_{n-j}}$ contains no nontrivial powers of $q$, i.e. $\phi_{n,x}=1$. 

Finally, if $x<\eta_n\le x+1$ then $\phi_{n,x+1}=1$ and $\phi_{n,x}=q^{\ell(\la)}$, thus $\phi_{n,x+1}-\phi_{n,x}\ge 1-q$. 

This implies all the assumptions of Proposition \ref{pr:equiv} and completes the proof. 
\end{proof}

We are now in a position to apply the above statements to the Macdonald measures. 

\begin{corollary}\label{cr:equiv-macd} For any sequence of specifications 
$\{\mathcal S^\MM_m\}_{m\ge 1}$  of the Macdonald measure with no nonzero 
$\beta$-parameters (see Definition \ref{df:mac-spec}), the random variables 
$\ell(\la)$ for any two pairs of parameters 
$(q_1,t_1),(q_2,t_2)\in [0,1)^2$ are asymptotically equivalent as $m\to\infty$. 

Similarly, for any two sequences of specifications 
$\{\mathcal S^\MM_m\}_{m\ge 1}$, $\{\widetilde{\mathcal S}^\MM_m\}_{m\ge 1}$ of 
the Macdonald measures whose $t$-parameters are the same, 
and whose $q$-parameters and $\beta$-parameters are related as follows:
$$
\tilde q=q^{1/k}, \qquad \{\tilde\beta_j\}=\{\beta_j\}\sqcup\{q^{1/k}\beta_j\}
\sqcup\dots\sqcup 
\{q^{(k-1)/k}\beta_j\},
$$
with some $k\in\{1,2,\dots\}$, the random variables $\ell(\la)$ are asymptotically equivalent as $m\to\infty$.
\end{corollary}
\begin{proof} Let us start with the second part. As was mentioned in Remark \ref{rm:independence}, if the $t$-parameter is the same then the expectations of the form \eqref{eq:macd-exp} do not change under the above replacements. Taking a generating function of these averages for a fixed $n$ and dividing by $\prod_{j\ge 0} (1+\zeta t^j)$ leads to the observable in the 
right-hand side of \eqref{eq:match3}, which does not change as well. 
Now for $t>0$ we may substitute $\zeta=t^x$, which gives $\phi_{n,x}$ of Proposition \ref{pr:column}, and the statement directly follows from that proposition. For $t=0$ one needs first to take the limit of both sides of \eqref{eq:macd-exp} normalized by $t^{\ell(\ell-1)/2}$ as $t\to 0$. This is done in \cite[Proposition 3.1.3]{BC}, and the result is an integral representation for $\E q^{\la_n+\dots+\la_{n-\ell+1}}$ that is also independent of the $(q,\beta)$-replacements in our hypothesis. Applying Proposition \ref{pr:t0} we obtain the result for $t=0$.

Let us proceed to the first part. In the absence of the $\beta$-parameters, all the expectations we just discussed are actually independent of $q$. Thus, we can change $q$ freely ($t$ is so far fixed), and the resulting sequences of Macdonald measures will have asymptotically equivalent $\ell(\la)$. Then we can choose $q=t$, when the Macdonald measures turn into the similarly specialized Schur measures that are actually independent of $q$ and $t$. Hence, for any initial pair $(q,t)\in [0,1)$, the random variable $\ell(\la)$ is asymptotically equivalent to the same random variable for the similarly specialized sequence of the Schur measures, and the result follows. 
\end{proof} 

\begin{remark} One instance of the asymptotic equivalence of Corollary 
\ref{cr:equiv-macd} is \cite[Theorem 1.3]{Dim}, cf. a discussion in \S1.4 
there. 
\end{remark}

Another application is an asymptotic equivalence of observables between the stochastic six vertex model and the Macdonald measures. 

\begin{corollary}\label{cr:equiv-6v-macd} Assume we are given two sequences of specifications $\{\mathcal S_m^{\vv}\}_{m\ge 1}$ and $\{\mathcal S^\MM_m\}_{m\ge 1}$, and assume that these specifications match for large enough $m$, cf. Definitions \ref{df:6v-spec}, \ref{df:mac-spec}, \ref{df:match}. Then the random variables $\HT(M,N)$ and $n-\ell(\la)$, defined for the vertex model and for the Macdonald measures, respectively, are asymptotically equivalent as $m\to\infty$.  
\end{corollary}
\begin{remark} The index $m$ in the above statement is essentially a placeholder for some limit transition in the space of parameters of matching specializations, and it can be made continuous if needed. For example, in the next section we will send $M$ and $N$ to infinity 
with all other parameters being fixed. The statement of Corollary \ref{cr:equiv-6v-macd}  becomes meaningful only if the mentioned random variables spread under the limit transition.  
\end{remark}
\begin{proof}[Proof of Corollary \ref{cr:equiv-6v-macd}] This is a corollary of \eqref{eq:match3}. The asymptotic equivalence of the right-hand side with $\zeta=t^x$ and the random variable $\ell(\la)-n$ was just discussed in the proof of Corollary \ref{cr:equiv-macd}. The asymptotic equivalence of the left-hand side with $\zeta=Q^x$ and the random variable $-\HT(M,N)$ follows from Proposition \ref{pr:equiv} and Example \ref{ex:function}, because
$$
\prod_{i\ge 0} \frac{1}{1+\zeta Q^{\HT(M,N)+i}}= \prod_{i\ge 0} \frac{1}{1+Q^{x+\HT(M,N)+i}},
$$
and the function $\Phi(x)=\prod_{i\ge 0}(1+Q^{x+i})^{-1}$ fits the format of Example \ref{ex:function}. Using \eqref{eq:match3} and changing the signs of the observables yields the desired statement. 
\end{proof} 

\section{Tracy-Widom asymptotics for homogeneous vertex models}\label{sc:asymptotics}

The goal of this section is to derive the GUE Tracy-Widom asymptotics for height function of the homogeneous vertex models described in Example \ref{ex:homog} using the connection to the Macdonald (or rather Schur) measures. 

\begin{theorem}\cite[Theorems 1.1 and 1.2]{BCG}\label{th:6v} Consider the stochastic homogeneous six vertex model in the quadrant, that is, $s_x\equiv Q^{-1/2}$,  $\xi_i\equiv 1$, $u_i\equiv u>0$ in the notation of Section \ref{sc:6v}. Denote $\zeta=Q^{-1/2}u^{-1}$ and note that $0<\zeta<1$, cf. Example \ref{ex:homog}. 
Then for any $\mu,\nu>0$ we have the following convergence in probability:
$$
\lim_{L\to\infty} \frac{\HT(\mu L,\nu L)}{L}=\mathfrak H(\mu,\nu),
$$
where 
$$
\mathfrak H(\mu,\nu)=\begin{cases} \dfrac{(\sqrt{\nu}-\sqrt{\zeta\mu})^2}{1-\zeta},&\zeta\le {\mu}/{\nu}\le \zeta^{-1},\\
0,& \mu/\nu\ge \zeta^{-1},\\
\nu-\mu,&\mu/\nu\le \zeta.
\end{cases}
$$
Furthermore, for $\zeta< \mu/\nu< \zeta^{-1}$ we have
$$
\lim_{L\to\infty}\Prob\left\{\frac{\HT(\mu L,\nu L)-\mathfrak H(\mu,\nu)L}{\sigma_{\mu,\nu}L^{1/3}}\ge -x\right\}=F_{\textrm{GUE}}(x),
$$
where $F_{\textrm{GUE}}$ is the GUE Tracy-Widom distribution, and 
$$
\sigma_{\mu,\nu}=\frac{\left(\zeta\mu\nu\right)^{1/6}\left(1-\sqrt{\zeta\mu/\nu}\right)^{2/3}\left
(1-\sqrt{\zeta\nu/\mu}\right)^{2/3}}{1-\zeta}\,.
$$
\end{theorem} 
\begin{proof} It suffices to consider the case of $\zeta<\mu/\nu<\zeta^{-1}$ because the freezing of the random path configuration outside this region follows from the following obvious properties of the height function
$$
0\le \HT(M_2,N)- \HT(M_1,N)\le M_1-M_2 \quad \textrm{for}\quad M_1\ge M_2, \qquad 0\le \HT(M,N)\le N,
$$ 
and from the fact that its law of large numbers $\mathfrak H(\mu,\nu)$ converges to the minimal and maximal possible values of 0 and $\nu-\mu$ at the edges of the region. 
Thus, from now on we will assume that $\zeta<\mu/\nu<\zeta^{-1}$. 

We will rely on Corollary \ref{cr:equiv-6v-macd} and instead prove a similar limiting statement for the length $\ell(\la)$ of the random Young diagram distributed according to the corresponding Macdonald measure. Since the variance of this random variable will tend to $\infty$, and the limiting distribution function $F_{GUE}(s)$ is continuous, we will conclude that $\ell(\la)$ spreads, and hence by Corollary \ref{cr:equiv-6v-macd} we will have the same convergence for $\HT(M,N)$. 

The matching specification of the Macdonald measure is described in Example \ref{ex:homog}: For $\HT(M,N)$ we can consider the Schur measure with 
$$
\Prob\{\la\}=const\cdot s_\la(\underbrace{\zeta u,\dots,\zeta u}_{M-1})s_\la(\underbrace{u^{-1},\dots,u^{-1}}_{N})=const\cdot s_\la(1^{M-1}) s_\la(1^N)\zeta^{|\la|},
$$
where the last equality is due to homogeneity of the Schur polynomials and the fact that $\deg s_\la=|\la|$. 

Asymptotic analysis of the Schur measures is a very well developed subject, see e.g. \cite{BG-lectures} and references therein. The key fact is that for the random partition $\la$, the random point configuration $\{\la_i-i\}_{i=1}^\infty\subset\Z$ generates a \emph{determinantal point process} (see e.g. \cite{B-det} and references therein for the general information on the latter). For a generic Schur measure this was first proved in \cite{Ok-schur}, where a convenient double contour integral formula for the corresponding correlation kernel was also derived. The particular case of the Schur measures above was actually considered a bit earlier in \cite{Joh-shape} and \cite{BO-hyper}, where they were also identified as the orthogonal polynomial ensembles associated with the Meixner classical orthogonal polynomials. 

The double contour integral formula of \cite{Ok-schur} for the correlation kernel describing the random configuration $X(\la):=\{\la_i-i\}_{i\ge 1}$ in our case above takes the form
\begin{equation}\label{eq:kernel}
K(x,y)=\frac{1}{(2\pi\i)^2}\oint\oint \frac{(1-\sqrt{\zeta}z^{-1})^N}{(1-\sqrt{\zeta}z)^{M-1}} \frac{(1-\sqrt{\zeta}w)^{M-1}}{(1-\sqrt{\zeta}w^{-1})^{N}} \frac{dzdw}{(z-w)z^{x+1}w^{-y}}
\end{equation}
with $x,y\in\Z$, and the integration contours being positively oriented circles satisfying $\zeta^{-1}>|z|=r_1>1>r_2=|w|>\zeta$.

We are interested in the behavior of $\ell(\la)$, and it is easy to see that $-\ell(\la)$ is the leftmost particle of the complementary point configuration $Y(\la):=\Z\setminus X(\la)$. Kerov's complementation principle for determinantal point processes, see \cite[A.3]{BOO}, states that $Y(\la)$ also generates a determinantal point process with the correlation kernel $\tilde K(x,y):=\mathbf{1}_{x=y}-K(x,y)$. Noting that $\Res_{z=w}$ of the integrand is exactly $\mathbf{1}_{x=y}$, we see that $\tilde K(x,y)$ is given by the same integral with interchanged contours, and with the minus sign in front (or with $(z-w)$ replaced by $(w-z)$). 

The inclusion-exclusion principle allows one to identify the \emph{gap 
probabilities} (equivalently the probabilities of not having any particles in 
a subset called ``gap'') for a determinantal point process as Fredholm 
determinant expansions for $\mathbf{1}$ minus the correlation kernel restricted 
to the gap. 
In our case, this means that $\Prob\{-\ell(\la)>x\}=\det(\mathbf{1}-\tilde K)_{\ell^2(x,x-1,x-2,\dots)}$. 

Finally, we need to perform asymptotic analysis of the kernel $\tilde K$ to see 
what the above Fredholm determinant converges to. Double contour integral 
representations provide a very convenient tool for such an analysis; this was 
first done in \cite{OkR}, \cite{Ok2003}. The reason is that the part of the 
integrand that depends on the large parameter $L$ can be written in the form
$\exp(L(G(z)-G(w))$, where in our case, cf. \eqref{eq:kernel},
$$
G(z)=-\mu\ln\left(1-\sqrt{\zeta}z\right)+\eta\ln\left(1-\sqrt{\zeta}z^{-1}
\right)-\frac xL\ln z.
$$
One can then try to deform the integration contours to the domains where $\Re 
G(z)<0$ and $\Re G(w)>0$, which would lead to a fast decay of the integral. 
Along the way the contours may need to cross or to come close to a common point;
in the first case the limit of the kernel is the residue at $z=w$ integrated 
over the parts of the contours that crossed, while in the second case the 
limiting behavior is determined from an infinitesimal neighborhood of the 
common point. 

The endpoints of the integration arc for the residue at $z=w$ 
end up being the saddle points of $\Re G(z)$, which are the critical points of 
$G(z)$. A common point for the contours arises when such critical points 
merge, and it is a double critical point of $G(z)$. The first case corresponds
to the values of $x$ near which the density of points in our point process is 
strictly between 0 and 1, the so-called bulk of the point process, while the 
second case corresponds to the edges of the bulk. 

In addition, if the contours can be deformed to the desired domains without 
getting close, the kernel tends to zero, and we would see no particles near 
such a location $x$, while if the deformation requires the residue at $z=w$ to
be taken on the whole closed contour, the kernel tends to $\mathbf{1}$, and 
almost all locations near such an $x$ are occupied by particles with high 
probability. 

This strategy has been worked out in dozens of papers and is completely 
standard by now. One could e.g. look at \cite{OkR-skew}, 
\cite{BF2008} for detailed examples. We will thus omit 
the usual arguments that the contours can be deformed to the needed positions 
(they are rather similar to the above references) and will focus instead on the 
critical point computation that will provide us with the final answer.

Since we are interested in the leftmost particle of the random point 
configuration $Y(\la)$, we need to investigate the edges of the system. 
Looking for values of $x\in\R$ that would lead to double critical points of 
$G(z)$ yields two values $x=(\sqrt{\zeta\mu}\pm\sqrt{\nu})^2/(1-\zeta)-\nu$. These 
correspond to the two edges of the bulk consisting of a single interval; since
we need the left edge we pick the smaller value, call it $x_c$. The corresponding double 
critical point is at
$$
z_c=\frac{\sqrt{\zeta\mu}-\sqrt{\nu}}{\sqrt{\mu}-\sqrt{\zeta\nu}}\,.
$$
We set 
$$
\sigma_{\mu,\nu}=-z_c\left(\frac{G'''(z_c)}{2}\right)^{1/3}=\frac{\left(\zeta\mu\nu\right)^{1/6}\left(1-\sqrt{\zeta\mu/\nu}\right)^{2/3}\left
(1-\sqrt{\zeta\nu/\mu}\right)^{2/3}}{1-\zeta}
$$
and observe that the substitution 
$$
x=x_c L-\sigma_{\mu,\nu} L^{1/3}\tilde x,\qquad
z=z_c\left(1+\frac{L^{-1/3}}{\sigma_{\mu,\nu}}\tilde z\right)  
$$
leads to
$$
G(z)=G(z_c)L+\ln(z_c)\sigma_{\mu,\nu}L^{2/3}\tilde x-\frac{\tilde z^3}{3}+\tilde x\tilde z. 
$$
Making a similar substitution for the second integration variable $w$, we conclude that
\begin{equation}
\label{eq:kernel-asympt}
\lim_{L\to\infty}\frac{e^{\ln(z_c)\sigma_{\mu,\nu}L^{2/3}\tilde y}}{e^{\ln(z_c)\sigma_{\mu,\nu}L^{2/3}\tilde x}}\cdot\sigma_{\mu,\nu}L^{1/3}\cdot\tilde K\left(x_cL-\sigma_{\mu,\nu} L^{1/3}\tilde x,x_cL-\sigma_{\mu,\nu} L^{1/3}\tilde y\right)
=K_{Airy}(\tilde x,\tilde y),
\end{equation}
where
$$
K_{Airy}(\tilde x,\tilde y)=\frac{1}{(2\pi\i)}\iint e^{{\tilde w^3}/{3}-{\tilde z^3}/{3}-\tilde w\tilde y+\tilde z\tilde x}\, \frac{d\tilde zd\tilde w}{\tilde w-\tilde z}=\frac{\textrm{Ai}(\tilde x)\textrm{Ai}'(\tilde y)-\textrm{Ai}'(\tilde x)\textrm{Ai}(\tilde y)}{\tilde x-\tilde y}
$$
is the Airy kernel. Here the $\tilde z$-contour goes from $ e^{4\pi i/3}\infty$ to $e^{2\pi i/3}\infty$ and the $\tilde w$-contour from $e^{5\pi i/3}\infty$ to $e^{\pi i/3}\infty$ so that the contours do not intersect, and $\textrm{Ai}(\,\cdot\,)$ is the Airy function. Note that the first prefactor of $\tilde K$ in \eqref{eq:kernel-asympt} plays no role as it does not affect $\det(1-\tilde K)$, and the second prefactor $\sigma_{\mu,\nu}L^{1/3}$ is responsible for the change of scale in the space where the point configurations live. 

Such a contour deformation argument proves that the limiting relation \eqref{eq:kernel-asympt} holds uniformly in $\tilde x$ and $\tilde y$ varying over a compact set in $\R$, and one needs a bit more to prove that $\det(\mathbf{1}-\tilde K)_{\ell^2(x,x-1,\dots)}$ converges to the GUE Tracy-Widom distribution 
$F_{GUE}(\tilde x)=\det(\mathbf{1}-K_{Airy})_{L^2(\tilde x,+\infty)}$ (see \cite{TW} for the latter). 
A straightforward approach consists in proving that the corresponding Fredholm determinant expansions converge, but this requires careful tail estimates of the contour integrals. In our particular case the situation is simpler, because our kernel $\tilde K=\textbf{1}-K$ becomes self-adjoint after a ``gauge transformation'' of the form $\tilde K(x,y)\mapsto f(x)\tilde K(x,y)/f(y)$ for an appropriate function $f$ (the first factor in \eqref{eq:kernel-asympt} is a remnant of such a conjugation). This follows from the fact that $K$ has the same property, and the corresponding self-adjoint kernel is the Christoffel-Darboux kernel for the Meixner orthogonal polynomials, see \cite{BO-Meixner} for details. For self-adjoint kernels the convergence of determinants can be reduced to the uniform convergence of kernels on compact sets plus the convergence of traces, see \cite[A.4]{BOO}. But the trace of $\tilde K$ can be computed explicitly by summing a geometric series in the integrand of \eqref{eq:kernel-asympt}, and 
the convergence to the corresponding quantity for the Airy kernel immediately follows from the same contour deformation argument.
\end{proof}

\begin{remark} The condition $\zeta<\mu/\nu<\zeta^{-1}$ that we imposed in the beginning of the proof above, is necessary to ensure that the edge point $x_cL$ indeed captures the behavior of $\ell(\la)$. By similar contour deformations one can show that if the above equalities are not satisfied, the leftmost point of $Y(\la)=\Z\setminus\{\la_i-i\}_{i\ge 1}$ is actually at $-\ell(\la)=-\min(M-1,N)$ with high probability, and the particle density of $Y(\la)$ in $[-\min(M-1,N),x_cL]$ is close to 1. which corresponds to the Young diagram of $\la$ developing a flat part of the boundary thanks to $\la_{\ell(\la)}\sim \min(M-1,N)-x_cL\to\infty$. 
\end{remark}

\begin{theorem} Consider the stochastic homogeneous higher spin six vertex 
model in the quadrant with parameters $s_x\equiv -Q^{1/2}$,  $\xi_i\equiv 1$, 
$u_i\equiv u>0$ in the notation of Section \ref{sc:6v}. Denote 
$\zeta=Q^{-1/2}u^{-1}$. 
Then for any $\mu,\nu>0$ we have the following convergence in probability
$$
\lim_{L\to\infty} \frac{\HT(\mu L,\nu L)}{L}=\mathfrak H(\mu,\nu),
$$
where 
$$
\mathfrak H(\mu,\nu)=\begin{cases} \dfrac{(\sqrt{\nu}-\sqrt{\zeta\mu})^2}{1+\zeta},& 0<{\mu}/{\nu}\le \zeta^{-1},\\
0,& \mu/\nu\ge \zeta^{-1}.
\end{cases}
$$
Furthermore, for $\mu/\nu< \zeta^{-1}$ we have
$$
\lim_{L\to\infty}\Prob\left\{\frac{\HT(\mu L,\nu L)-\mathfrak H(\mu,\nu)L}{\sigma_{\mu,\nu}L^{1/3}}\ge -x\right\}=F_{\textrm{GUE}}(x),
$$
where $F_{\textrm{GUE}}$ is the GUE Tracy-Widom distribution, and 
$$
\sigma_{\mu,\nu}=\frac{\left(\zeta\mu\nu\right)^{1/6}\left(1-\sqrt{\zeta\mu/\nu}\right)^{2/3}\left
(1+\sqrt{\zeta\nu/\mu}\right)^{2/3}}{1+\zeta}\,.
$$
\end{theorem} 
\begin{proof} The argument very closely follows the proof of Theorem \ref{th:6v} so we will only point out the differences. The corresponding Schur measure now has the form
$$
\Prob\{\la\}=const\cdot s_{\la'}(1^{M-1}) s_\la(1^N)\zeta^{|\la|},
$$
where $\la'$ is the dual partition to $\la$ (their Young diagrams are transposed to each other). The correlation kernel has the form
\begin{equation*}
K(x,y)=\frac{1}{(2\pi\i)^2}\oint\oint \frac{(1+\sqrt{\zeta}z)^{M-1}(1-\sqrt{\zeta}z^{-1})^N}{(1+\sqrt{\zeta}w)^{M-1}(1-\sqrt{\zeta}w^{-1})^{N}} \frac{dzdw}{(z-w)z^{x+1}w^{-y}},
\end{equation*} 
and it is related to the Christoffel-Darboux kernel of the Krawtchouk classical orthogonal polynomials. The needed asymptotic analysis of this kernel has actually been done in \cite{Joh-Ann}, \cite{Joh-Aztec}, but it is simpler for us to follow the same line of reasoning rather than to match the notation. The new function $G(z)$ has the form
$$
G(z)=\mu\ln\left(1+\sqrt{\zeta}z\right)+\eta\ln\left(1-\sqrt{\zeta}z^{-1}
\right)-\frac xL\ln z,
$$
and the needed lower edge and the corresponding double critical point are
$$
x_c=(\sqrt{\nu}-\sqrt{\zeta\mu})^2/(1+\zeta)-\nu,\qquad z_c=\frac{\sqrt{\zeta\mu}-\sqrt{\nu}}{\mu+\sqrt{\zeta\nu}}\,. 
$$
The fluctuation scale $\sigma_{\mu,\nu}$ is now given by 
$$
\sigma_{\mu,\nu}=-z_c\left(\frac{G'''(z_c)}{2}\right)^{1/3}=\frac{\left(\zeta\mu\nu\right)^{1/6}\left(1-\sqrt{\zeta\mu/\nu}\right)^{2/3}\left
(1+\sqrt{\zeta\nu/\mu}\right)^{2/3}}{1+\zeta},
$$ 
and the rest of the proof is exactly the same. 
\end{proof}

\end{document}